\documentclass[12pt,oneside,english]{amsart}
\usepackage[T1]{fontenc}
\usepackage[latin9]{inputenc}
\usepackage{geometry}
\usepackage{amsthm}
\usepackage{amssymb}
\usepackage{graphicx}
\usepackage{esint}

\makeatletter

\let\SF@@footnote\footnote
\def\footnote{\ifx\protect\@typeset@protect
    \expandafter\SF@@footnote
  \else
    \expandafter\SF@gobble@opt
  \fi
}
\expandafter\def\csname SF@gobble@opt \endcsname{\@ifnextchar[%]
  \SF@gobble@twobracket
  \@gobble
}
\edef\SF@gobble@opt{\noexpand\protect
  \expandafter\noexpand\csname SF@gobble@opt \endcsname}
\def\SF@gobble@twobracket[#1]#2{}

\numberwithin{equation}{section}
\numberwithin{figure}{section}
\theoremstyle{plain}
\newtheorem{thm}{\protect\theoremname}
  \theoremstyle{definition}
  \newtheorem{defn}[thm]{\protect\definitionname}

\makeatother

\usepackage{babel}
  \providecommand{\definitionname}{Definition}
\providecommand{\theoremname}{Theorem}

\begin{document}

\title{Population Stability and Momentum }

\maketitle

\section*{Arni S.R. Srinivasa Rao%
\footnote{
\dedicatory{1. To commemorate MPE2013 launched by the International Mathematics
Union, the author dedicates this work to Alfred J. Lotka.}\maketitle
}$^{,}$%
\footnote{
\dedicatory{Arni S.R. Srinivasa Rao is an Associate Professor at Georgia Regents
University, Augusta, USA. His email address is arrao@gru.edu }\maketitle
}$ $$^{,}$%
\footnote{
\dedicatory{Acknowledgements: Very useful comments by referees have helped to
re-write some parts of the article for better readability.  Comments
from Professor JR Drake (University of Georgia, Athens), NV Joshi
(Indian Institute of Science, Bangalore) and from Dr. Cynthia Harper
(Oxford) helped very much in exposition of the paper. My sincere gratitude
to all. Author is partially supported by DFID and BBSRC, UK in the
form of CIDLID project BB/H009337}\maketitle
}}

(\textbf{Appeared in }\textbf{\emph{Notices of the American Mathematical
Society)}}
\begin{abstract}
A new approach is developed to understand stability of a population
and further understanding of population momentum. These ideas can
be generalized to populations in ecology and biology.
\end{abstract}

\section{Introduction}

One commonly prescribed approach for understanding the stability of
system of dependent variables is that of Lyapunov. In a possible alternative
approach - when variables in the system have momentum then that can
trigger additional dynamics within the system causing the system to
become unstable. In this study stability of population is defined
in terms of elements in the set of births and elements in the set
of deaths. Even though the cardinality of the former set has become
equal to the cardinality of the latter set, the momentum with which
this equality has occurred determines the status of the population
to remain at \emph{stable.} Such arguments also works for the other
population ecology problems.

\section{Population Stability Theory}

Suppose $\left|P_{N}(t_{0})\right|$ be the cardinality of the set
of people, $P_{N}(t_{0}),$ representing population at global level
at time $t_{0}$, where $P_{N}(t_{0})=\left\{ u_{1},u_{2},\cdots,u_{N}\right\} $,
the elements $u_{1},u_{2},\cdots,u_{N}$ represent individuals in
the population. Broadly speaking, the Lyapunov stability principles
(see \cite{VLL}) suggests, $ $$\left|P_{N}(t_{0})\right|$ is asymptotically
stable at population size $U$, if $\left|\left|P_{N}(T)\right|-U\right|<\epsilon$
($\epsilon>0)$ at all $T$ whenever $T>t_{0}$. In some sense, $\left|P_{N}(t_{0})\right|$
attains the value $U$ over the period of time. Lotka-Voltera's predator
and prey population models provide one of the classical and earliest
stability analyses of population biology (see for example, \cite{JDM})
and Lyapunov stability principles often assist in the analysis of
such models. These models have equations that describe the dynamics
of at least two interacting populations with parameters describing
interactions and natural growth. Outside human population models and
ecology models, stability also plays a very important role in understanding
epidemic spread \cite{AR}. In this paper, we are interested in factors
that cause dynamics in $P_{N}$ and relate these factors with status
of stability. A set of people $P_{M}(t_{0})=\left\{ u_{m_{1}},u_{m_{2}},\cdots,u_{m_{M}}\right\} $,
where $P_{M}(t_{0})\subset P_{N}(t_{0})$, are responsible for increasing
the population (reproduction) during the period $(t_{0},s)$ and contribute
to $P_{N}(s)$, the set of people at $s$ (if they survive until the
time $s$). The set $Q_{M_{1}}(s-t_{0})=\left\{ v_{M_{11}},v_{M_{12}},\cdots,v_{M_{1M_{1}}}\right\} $
represent removals (due to deaths) from $P_{N}(t_{0})$ during the
time interval $(t_{0},s)$. Let $\mbox{R}_{\phi}(s-t_{0})$ be the
period reproductive rate (net) applied on $P_{M}(t_{0})$ for the
period $(t_{0},s)$, then the number of new population added during
$(t_{0},s)$ is $\mbox{R}_{\phi}(s-t_{0})\left|P_{M}(t_{0})\right|.$
Net reproduction rate at time $t_{0}$ (or in a year $t_{0}$) is
the average number of female children that would be born to single
women if she passes through age-specific fertility rates and age-specific
mortality rates that are observed at $t_{0}$ (for the year $t_{0}$).
Since net reproductive rates are futuristic measures, we use period
(annual) reproductive rates for computing period (annual) increase
in population. Let $\mbox{C}_{N_{1}}(s-t_{0})=\left\{ w_{1},w_{2},\cdots,w_{N_{1}}\right\} $
be the set of newly added population during $(t_{0},s)$ to the set
$P_{N}(t_{0}).$ After allowing the dynamics during $(t_{0},s)$,
the population at $s$ will be

\begin{eqnarray}
P_{N}(t_{0})\cup C_{N_{1}}(s-t_{0})-Q_{M_{1}}(s-t_{0}) & = & \begin{cases}
\left.\begin{array}{c}
u:u\in P_{N}(t_{0})\cup C_{N_{1}}(s-t_{0})\mbox{ }\\
\mbox{and }u\notin Q_{M_{1}}(s-t_{0})
\end{array}\right\} \end{cases}\nonumber \\
 & = & \left\{ u_{1},u_{2},\cdots,u_{N+N_{1}-M_{1}}\right\} \label{eq:1}
\end{eqnarray}

Note that, $Q_{M_{1}}(s-t_{0})\subset P_{N}(t_{0})\cup C_{N_{1}}(s-t_{0})$,
because the set of elements $\left\{ v_{M_{1}},v_{M_{12}},\right.$
$\left.\cdots,v_{M_{1M_{1}}}\right\} $ eliminated during the time
period $(t_{0},s)$ are part of the set of elements 

$\left\{ u_{1},u_{2},\cdots,u_{N},w_{1}\right.$ $\left.,w_{2},\cdots,w_{N_{1}}\right\} $
and the resulting elements surviving by the time $s$ are represented
in equation (\ref{eq:1}). The element $u_{1}$ in the set (\ref{eq:1})
may not be the same individual in the set $P_{N}(t_{0})$. Since we
wanted to retain the notation that represents people living at each
time point, so for ordering purpose, we have used the symbol $u_{1}$
in the set (\ref{eq:1}).

Using Cantor\textendash{}Bernstein\textendash{}Schroeder theorem \cite{PS},
$ $$\left|\mbox{C}_{N_{1}}(s-t_{0})\right|=$ $\left|Q_{M_{1}}(s-t_{0})\right|$
if $ $$\left|\mbox{C}_{N_{1}}(s-t_{0})\right|$ $\leq$ $\left|Q_{M_{1}}(s-t_{0})\right|$
and $ $$\left|\mbox{C}_{N_{1}}(s-t_{0})\right|$ $\geq\;$ $ $ $\left|Q_{M_{1}}(s-t_{0})\right|.$
If $ $$\left|\mbox{C}_{N_{1}}(s-t_{0})\right|=$ $\left|Q_{M_{1}}(s-t_{0})\right|$
then the natural growth of the population (in a closed situation)
is zero and if this situation continues further over the time then
the population could be termed as stationary. Assuming these two quantities
are not same at $t_{0}$, the process of two quantities $\left|\mbox{C}_{N_{1}}\right|$
and $\left|Q_{M_{1}}\right|$ becoming equal could eventually happen
due to several sub-processes. 

Case I: $\left|\mbox{C}_{N_{1}}\right|>$ $\left|Q_{M_{1}}\right|$
at time $t_{0}.$ We are interested in studying the conditions for
the process $\left|\mbox{C}_{N_{1}}\right|\rightarrow$ $\left|Q_{M_{1}}\right|$
for some $s>t_{0}$. Two factors play a major role in determining
the speed of this process, they are, compositions of the family of
sets $\left[\left\{ P_{M}(s)\right\} :\forall s>t_{0}\right]$ and
$\left[\left\{ \mbox{R}_{\phi}(s)\right\} :\forall s>t_{0}\right]$.
Suppose $\left|P_{M}(s_{1})\right|>\left|P_{M}(s_{2})\right|>\cdots>\left|P_{M}(s_{T})\right|$
but the family of $\left\{ \left|\mbox{R}_{\phi}(s)\right|\right\} $
does not follow any decreasing pattern for some $t_{0}<s_{1}<s_{2}<...<s_{T}<s$,
then $\left|\mbox{C}_{N_{1}}\right|\nrightarrow\left|Q_{M_{1}}\right|$
by the time $s_{T}$. If $\mbox{R}_{\phi}(s_{1})>\mbox{R}_{\phi}(s_{2})>...>\mbox{R}_{\phi}(s_{T})$
for $t_{0}<s_{1}<s_{2}<...<s_{T}<s$ such that $\left|\mbox{R}_{\phi}(s_{T}-s_{T-1})\left|P_{M}(s_{T-1})\right|-Q_{M_{1}}(s_{T}-s_{T-1})\right|\rightarrow0$
for some sufficiently large $T>t_{0}$ and sufficiently small $\left|\mbox{R}_{\phi}(s_{T}-s_{T-1})\right|$,
then $\left|\mbox{C}_{N_{1}}\right|\rightarrow$ $\left|Q_{M_{1}}\right|$
by the time $s_{T}.$ Note that in an ideal demographic transition
situation, both these quantities should decline over the period and
the rate of decline of $ $$\left|Q_{M_{1}}\right|$ is slower than
the rate of decline in $\left|\mbox{C}_{N_{1}}\right|$ because $\left|\mbox{C}_{N_{1}}\right|>$
$\left|Q_{M_{1}}\right|$ at time $t_{0}.$ Demographic transition
theory, in simple terms, is all about, determinants, consequences
and speed of declining of high rates of fertility and mortality to
low levels of fertility and mortality rates. For introduction of this
concept see \cite{KD} and for an update of recent works, see \cite{Caldwell}.
Above trend of $\left|P_{M}(s_{1})\right|>\left|P_{M}(s_{2})\right|>\cdots>\left|P_{M}(s_{T})\right|$
(i.e. decline in people of reproductive ages over the time after $t_{0}$)
happens when births continuously decrease for several years. Following
the trend $\mbox{R}_{\phi}(s_{1})>\mbox{R}_{\phi}(s_{2})>...>\mbox{R}_{\phi}(s_{T})$
will lead to decline in new born babies and this will indirectly result
in decline in rate of growth of people who have reproductive potential.
However the decline in $\left|\mbox{R}_{\phi}(s)\right|$ for $s>t_{0}$
is well explained by social and biological factors, which need not
follow any pre-determined mathematical model. However the trend in
$\left|\mbox{R}_{\phi}(t)\right|$ for $t<t_{0}$ can be explained
using models by fitting parameters obtained from data. During the
entire process the value of $\left|Q_{M_{1}}\right|$ after time $t_{0}$
is assumed to be dynamic and decreases further. If a population continues
to remain at this stage of replacement we call it a \emph{stable population}.
The cycle of births, population aging and deaths is a \emph{continuous
process} with \emph{discretely quantifiable factors}. Due to improvement
in medical sciences there could be some delay in deaths, but eventually
the aged population has to be moved out of $\left\{ P_{N}\right\} $,
and consequently, population stability status can be broken with a
continuous decline in $\left\{ \left|C_{N_{1}}\right|\right\} $. 

Case II: $\left|\mbox{C}_{N_{1}}\right|=$ $\left|Q_{M_{1}}\right|$
at time $t_{0}.$ It is important to ascertain whether this situation
was immediately proceeded by case I or case II before determining
the stability process. Suppose case II is immediately preceded by
case I, then the rapidity and magnitude at which the difference between
$\left|\mbox{C}_{N_{1}}\right|$ and $\left|Q_{M_{1}}\right|$ was
shrunk prior to $t_{0}$ need to be quantified. Let us understand
the contributing factors for the set $Q_{M}$. At each $t$, there
is a possibility that the elements from the sets $C_{N_{1}}$, $P_{N}-C_{N_{1}}-P_{M}$,
$P_{M}$ are contributing to the set $Q_{M}.$ Due to high infant
mortality rates, the contribution of $C_{N_{1}}$ into $Q_{M}$ is
considered to be high, deaths of adults of reproductive ages, $P_{M}$,
and all other individuals (including the aged), $P_{N}-C_{N_{1}}-P_{M}$,
will be contributing to the set $Q_{M}$. Case II could occur when
$\left|\mbox{C}_{N_{1}}\right|$ and $\left|Q_{M_{1}}\right|$ are
at higher values or at lower values. Equality at higher values possibly
indicates, the number of deaths due to three factors mentioned here
are high (including high old age deaths) and these are replaced by
equal high number of births, i.e. $\left|R_{\phi}\right|$ and $\left|P_{M}\right|$
are usually high to reproduce a high birth numbers. If equality at
lower values of $\left|\mbox{C}_{N_{1}}\right|$ and $\left|Q_{M_{1}}\right|$
occurs after phase of case I then the chance of $P_{N}$ remaining
in stable position is higher. Suppose elements of $P_{N}$ are arbitrarily
divided into $k-$independent and non-empty subsets, $A(1)$, $A(2)$,
$\cdots$,$A(k)$ such that $ $$\left|P_{N}\right|$$=$$\int_{1}^{k}\left|A(s)\right|ds$.
Let $F$ be the family of all the sets $A(s)$ such that $\cup\left(A(s\right))=P_{N}$.\textbf{
}Members of $F$ are disjoint. Suppose $\left(\begin{array}{c}
F\\
k^{*}
\end{array}\right)$ be an arbitrary size of $k^{*}$ of subset of $F$ are satisfying
the case II and $F-\left(\begin{array}{c}
F\\
k^{*}
\end{array}\right)$ are not satisfying at time $t_{0}$ and $t>t_{0}$, then we are not
sure of total population also attains stability by Theorem \ref{local vs global stability theorem}.

\begin{figure}
\includegraphics[scale=0.6]{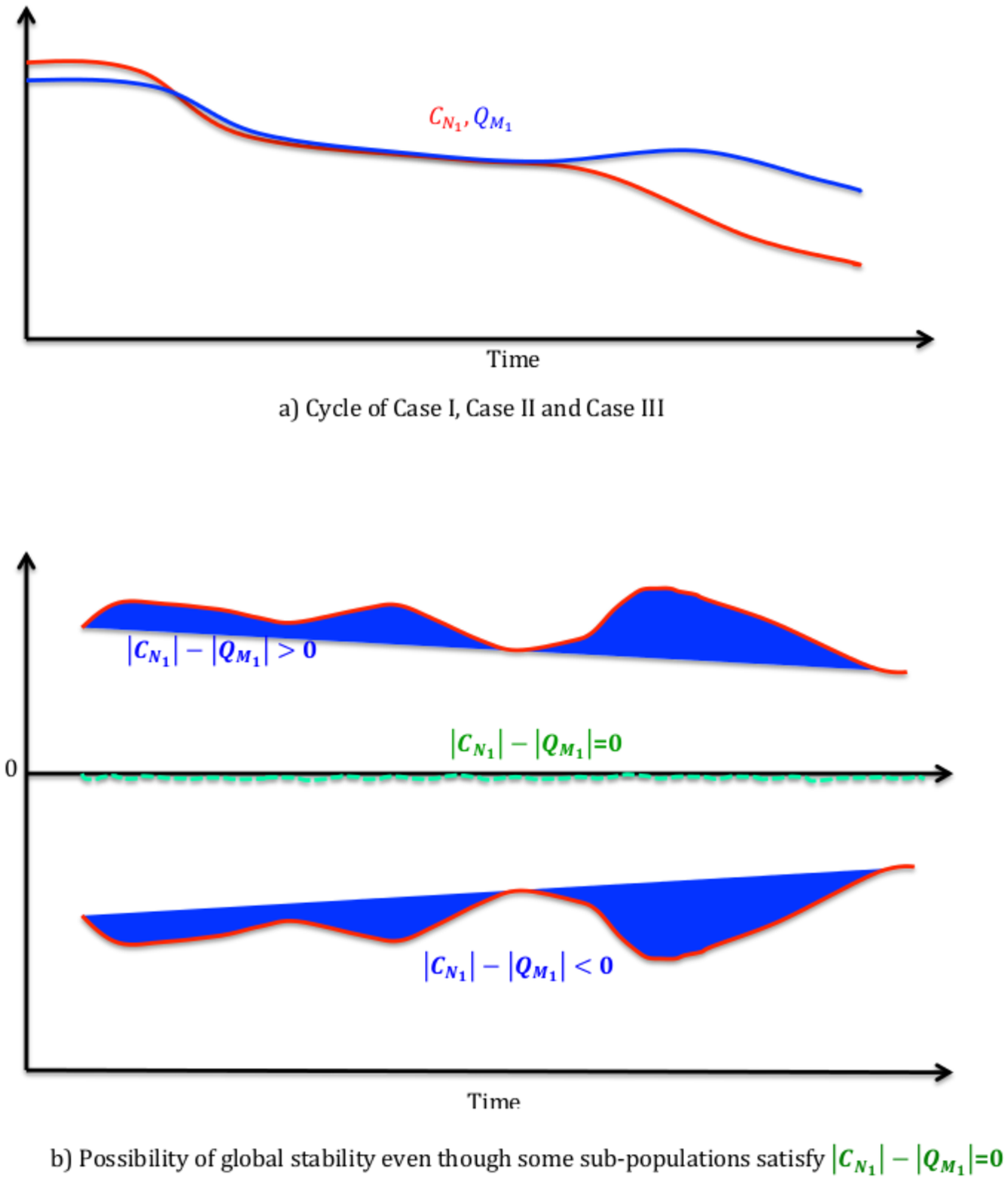}

\caption{(a) The cycle of all the cases could follow one after another and
the quantity at which equality of $C_{N_{1}}$ and $Q_{M_{1}}$ occurs
determines the duration of the case II. (b) Some of the sub-populations
which are not satisfied the equality of $C_{N_{1}}$ and $Q_{M_{1}}$
is compensated by the other sub-populations which are satisfying either
$C_{N_{1}}>$ $Q_{M_{1}}$ or $C_{N_{1}}$ $<Q_{M_{1}}$. }

\end{figure}

\begin{thm}
\label{local vs global stability theorem}Suppose each of the member
of $\left(\begin{array}{c}
F\\
k^{*}
\end{array}\right)$ is satisfying the condition $\left|\mbox{C}_{N_{1}}\right|=$ $\left|Q_{M_{1}}\right|$
and $F-\left(\begin{array}{c}
F\\
k^{*}
\end{array}\right)$ are not satisfying the condition $\left|\mbox{C}_{N_{1}}\right|=$
$\left|Q_{M_{1}}\right|$ at time $t\geq t_{0}$, then this does not
always leads $P_{N}$ to stability. \end{thm}
\begin{proof}
Note that $F$ has collection of $k-$sets. Suppose a collection $C$
divides $C_{N_{1}}$ into $k-$components of subpopulations $\left\{ C_{N_{1}}(1),C_{N_{1}}(2),\right.$
$\left.\cdots,C_{N_{1}}(k)\right\} $$ $ such that $ $$\left|C_{N_{1}}\right|$$=$$\int_{1}^{k}\left|C_{N_{1}}(s)\right|ds$,
where $C_{N_{1}}(s)$ is the $s^{th}-$ subset in $C$ and a collection
$Q$ divides $Q_{M_{1}}$ into $k-$ components of subpopulations
$\left\{ Q_{M_{1}}(1),Q_{M_{1}}(2),\cdots,Q_{M_{1}}(k)\right\} $
such that

$ $$\left|Q_{M_{1}}\right|$ $\;=\;$$\int_{1}^{k}\left|C_{M_{1}}(s)\right|ds$,
where $Q_{M_{1}}(s)$ is the $s^{th}-$subset in $Q$. 

By hypothesis, $\left|C_{N_{1}}(s^{*})\right|=\left|Q_{M_{1}}(s^{*})\right|$
for $s^{*}\in$ $\left\{ 1^{*},2^{*},\cdots,k^{*}\right\} $ at each
time $t\geq t_{0}$ until, say, $t_{T}$. The order between $k^{*}$
and $k-k^{*}$ could be one of the following: $2k^{*}<k$, $2k^{*}>k$,
$k^{*}=\frac{k}{2}$. Suppose $C_{N_{1}}\subset C$ and $Q_{M_{1}}\subset Q$
with 

\begin{eqnarray*}
C_{N_{1}}^{*} & = & \left\{ C_{N_{1}}^{*}(1),C_{N_{1}}^{*}(2),\cdots,C_{N_{1}}^{*}(k)\right\} \\
Q_{M_{1}} & = & \left\{ Q_{M_{1}}^{*}(1),Q_{M_{1}}^{*}(2),\cdots,Q_{M_{1}}^{*}(k)\right\} 
\end{eqnarray*}

for same above arbitrary combination of $k^{*}-$components and rest
of the $k-k^{*}$ components are satisfying $\left|C_{N_{1}}^{**}(s^{**})\right|-\left|Q_{M_{1}}^{**}(s^{**})\right|\neq0$
for all $s^{**}=1,2,\cdots,k-k^{*}$. We obtain unstable integral
over all $k-k^{*}$ components to ascertain the magnitude of unstability. 

$ $
\begin{eqnarray}
\int_{1}^{k-k^{*}}\left[\left|C_{N_{1}}^{**}(s^{**})\right|-\left|Q_{M_{1}}^{**}(s^{**})\right|\right]ds^{**}\label{destable-integral}\\
\nonumber 
\end{eqnarray}

The stable integral for this situation is 

\begin{eqnarray}
\int_{1}^{k^{*}}\left[\left|C_{N_{1}}^{*}(s^{*})\right|-\left|Q_{M_{1}}^{*}(s^{*})\right|\right]ds^{*}\label{stable-integral}\\
\nonumber 
\end{eqnarray}

To check the unstable and stable points over the time period $(t_{0},t_{T})$,
one can compute following integrals: 

\begin{eqnarray}
\int_{t_{0}}^{t_{T}}\int_{1}^{k-k^{*}}\left[\left|C_{N_{1}}^{**}(s^{**})\right|-\left|Q_{M_{1}}^{**}(s^{**})\right|\right]ds^{**}du\label{destable-integral over time period}\\
\nonumber 
\end{eqnarray}

\begin{eqnarray}
\int_{t_{0}}^{t_{T}}\int_{1}^{k^{*}}\left[\left|C_{N_{1}}^{*}(s^{*})\right|-\left|Q_{M_{1}}^{*}(s^{*})\right|\right]ds^{*}du\label{stable-integral-during time period}\\
\nonumber 
\end{eqnarray}

For each of the $k-k^{*}$ component, the values of $\left|C_{N_{1}}^{**}(s^{**})\right|-\left|Q_{M_{1}}^{**}(s^{**})\right|$
can be either positive or negative. If at time $t_{0}$, for all $s^{**}=1,2,\cdots,k-k^{*}$,
the values of $\left|C_{N_{1}}^{**}(s^{**})\right|-\left|Q_{M_{1}}^{**}(s^{**})\right|$
are positive (or negative) then the eq. (\ref{destable-integral})
will take a positive (or negative) quantity and the population at
time $t_{0}$ is not stable. If such a situation continues for all
$t_{T}\geq t_{0}$, then the integral in eq. (\ref{destable-integral over time period})
would never become zero and the population remains unstable in the
entire period $(t_{0},t_{T})$. However, for some of the $s^{**}$,
if the quantity $\left|C_{N_{1}}^{**}(s^{**})\right|-\left|Q_{M_{1}}^{**}(s^{**})\right|$
is positive and for other $s^{**}$, if the quantity $\left|C_{N_{1}}^{**}(s^{**})\right|-\left|Q_{M_{1}}^{**}(s^{**})\right|$
is negative such that eq. (\ref{destable-integral}) is zero at each
of the time points for the period $(t_{0},t_{T})$ then the population
remains stable during this period (because by hypothesis the eq. (\ref{stable-integral-during time period})
is zero). 
\end{proof}
Case III. $\left|\mbox{C}_{N_{1}}\right|<$ $\left|Q_{M_{1}}\right|$
at time $t_{0}.$ Global occurrence of this case at lower values of
$\left|\mbox{C}_{N_{1}}\right|$ and $\left|Q_{M_{1}}\right|$ indicates
that the $P_{N}$ is declining and also is in unstable mode. $R_{\phi}$
has been very low consistently for the period $t<t_{0}$ and the supply
to the set $P_{M}$ has diminished over a period in the past. All
the subsets of $C_{N_{1}}$ and $Q_{M_{1}}$ might not be stable in
case III, but by similar arguments of the Theorem \ref{local vs global stability theorem},
global population behavior nullifies some of the local population
and case III is still satisfied globally. 

All three cases would be repeated one following another. Most countries
are currently facing case I with varying distance between $\left|\mbox{C}_{N_{1}}\right|$
and $\left|Q_{M_{1}}\right|$ .

\section{Replacement Metric}

We introduce a metric, $d_{M}$, which we call \emph{a replacement
metric,} with a space, $M_{r}$ as follows:
\begin{defn}
\label{(Replacement-Metric)}(Replacement Metric). 

Let $A_{1}=\min\left\{ \left|\left|\mbox{C}_{N_{1}}(s)\right|-\left|Q_{M_{1}}(s)\right|\right|:s>t_{0}\right\} $
and 

$A_{2}=\max\left\{ \left|\left|\mbox{C}_{N_{1}}(s)\right|-\left|Q_{M_{1}}(s)\right|\right|:s>t_{0}\right\} $.
Let $M_{r}=\left[A_{1},A_{2}\right]\subset\mathbb{R}^{+}$ and $M=\left\{ \left|\left|\mbox{C}_{N_{1}}(s)\right|-\left|Q_{M_{1}}(s)\right|\right|\right.$
$\left.:s>t_{0}\right\} $ with the metric $d_{M}(x,y)=\frac{\left|x-y\right|}{2}$
. We can verify that $\left(M,d_{M}\right)$ is a metric space with
$d_{M}:\left(M\times M\right)\rightarrow M_{r}$ and non-empty set
$M.$
\end{defn}
The metric $M$, in the definition 1 is bounded, because $d_{M}\left(x,y\right)<k$
for $k>0.$ 
\begin{defn}
Suppose $\left|\left|\mbox{C}_{N_{1}}(s_{1})\right|-\left|Q_{M_{1}}(s_{1})\right|\right|=f_{1}$,
$\left|\left|\mbox{C}_{N_{1}}(s_{2})\right|\right.$ $-$$\left.\left|Q_{M_{2}}(s_{2})\right|\right|=f_{2}$
and so on for $s_{1}<s_{2}<...$ . Then we say population is stable
if $f_{s_{T}}\rightarrow0$ for sufficiently large $T$ and $ $$\frac{d}{ds_{T}}\left|\mbox{C}_{N_{1}}(s_{T})\right|=\frac{d}{ds_{T}}\left|\mbox{C}_{N_{1}}(s_{T})\right|=0$.
\end{defn}

\section{Conclusions}

We can prove that the value at which the population remains stable
is variable, i.e. the value at which the population becomes unstable
by deviating from case II could be different from the value (at a
future point in time) population becomes stable when it converges
to case II. \emph{Replacement metrics }(see definition \ref{(Replacement-Metric)})
are helpful in seeing this argument and such analysis is not possible
by Lotka-Voltera or Lyupunov methods. Due to population momentum,
there will be an increase in the population even though the reproduction
rate of the population becomes below the replacement level. Population
stability will always attain a local stable points before diverging
and again converging at a local stable point. The duration of a local
stable point depends on the density of the population and resources
available for the population.

\end{document}